\documentclass[10pt,twocolumn]{article}

\newif\ifTR\TRtrue

\usepackage{amsmath,amssymb,color,graphicx}
\usepackage{relsize}
\usepackage{IEEEtrantools}

\usepackage{algorithm}
\usepackage{algorithmic}

\usepackage[left=0.6in,top=0.7in,right=0.6in,bottom=0.5in,nohead,textheight=10in,footskip=0.3in,columnsep=0.3in]{geometry}

\usepackage{hyperref}

\usepackage{dsfont}

\newenvironment{proof}
    {\noindent{\em Proof.}
    }
    {
\raggedright{$\square$}
\vspace{3pt}
    }

\long\def\ignore#1{}

\def\myparagraph#1{\vspace{0pt}\noindent{\bf #1~~}}

\usepackage{xspace}
\def\PB{{\tt PB}\xspace}
\def\GPB{{\tt GPB}\xspace}
\def\CFG{{\tt CFG}\xspace}
\def\WCFG{{\tt WCFG}\xspace}
\def\CYK{{\tt CYK}\xspace}
\def\PBs{{\tt PB}s\xspace}
\def\GPBs{{\tt GPB}s\xspace}

\def\WCFGs{{\tt WCFG}s\xspace}
\def\CNF{{\tt CNF}\xspace}

\newcommand{\eqdef}{{\stackrel{\mbox{\tiny \tt ~def~}}{=}}}

\newtheorem{theorem}{Theorem}

\newtheorem{lemma}[theorem]{Lemma}

\begin{document}

\title{Combining pattern-based CRFs and weighted context-free grammars}
\author{
Rustem~Takhanov  \\ \tt{takhanov@ist.ac.at} 
\and
Vladimir~Kolmogorov \\ \tt{vnk@ist.ac.at}
}
\date{}
\maketitle

\begin{abstract}
We consider two models for the sequence labeling (tagging) problem.
The first one is a {\em Pattern-Based  Conditional Random Field }(\PB),
in which the energy of a string (chain labeling) $x=x_1\ldots x_n\in D^n$
is a sum of terms over intervals $[i,j]$ where each term is non-zero only if the substring $x_i\ldots x_j$
equals a prespecified word $w\in \Lambda$.
The second model is a {\em Weighted Context-Free Grammar }(\WCFG) frequently used for natural language processing.
\PB and \WCFG encode local and non-local interactions respectively, and thus can be viewed as complementary.

We propose a {\em Grammatical Pattern-Based CRF model }(\GPB) that combines the two  in a natural way.
We argue that it has certain advantages over existing approaches such
as the {\em Hybrid model} of Bened{\'i} and Sanchez that combines {\em $\mbox{$N$-grams}$} and \WCFGs.
The focus of this paper is to analyze the complexity of inference tasks in a \GPB such as computing MAP.
We present a polynomial-time algorithm for general \GPBs
and a faster version for a special case that we call {\em Interaction Grammars}.
\end{abstract}

\section{Introduction}
The {\em sequence labeling} (or the {\em sequence tagging}) problem is a supervised learning problem with the following formulation: given an observation
$z$ (which is usually a sequence of $n$ values), infer labeling $x=x_1\ldots x_n$
where each variable $x_i$ takes values
in some finite domain $D$.
Such problem appears in many domains such as text and speech analysis, signal analysis, and bioinformatics. Standard approaches to this problem include Hidden Markov Models ({\em HMM}) and Conditional Random Fields ({\em CRFs}).

In many applications  labelings $x$ satisfy the following ``sparsity'' assumption:
subwords $x_{i:j}\eqdef x_i\ldots x_j$ of a fixed length $k=j-i+1$
are distributed not uniformly over $D^k$, but are rather concentrated
in a small subset of $D^k$.
Words in this subset are called ``patterns''; we will denote the set of patterns as $\Lambda\subseteq D^\ast=\bigcup_{k\ge 0}D^k$.
Usually,  $\Lambda$ is taken as the set of short words (e.g.\ of length $k<5$)
that occur sufficiently often as subwords of labelings in the training data.
For problems satisfying this assumption it is natural to define a model given by the probability distribution $p_\theta(x|z)=\frac{1}{Z}\exp\{-E^{pb}_\theta(x|z)\}$
with the energy function 
\begin{equation}
E^{pb}_\theta(x|z)=\sum_{\alpha\in\Lambda}\sum_{\substack{[i,j]\subseteq[1,n]\\j-i+1=|\alpha|}}\psi^\alpha_{ij}\cdot[x_{i:j}=\alpha]
\label{eq:patternCRF}
\end{equation}
where $\theta$ is a vector of parameters to be learned from data, $\psi^\alpha_{ij}$ is a function that can depend on $z$ and $\theta$, $|\alpha|$ is the length of word $\alpha$ and $[\cdot]$ is the {\em Iverson bracket} (i.e. $[s]=1$ if statement $s$ is true, otherwise $[s]=0$). This model is called a {\em pattern-based CRF }(\PB) \cite{Ye:NIPS09,TK:ICML}.

Intuitively, pattern-based CRFs allow to model long-range interactions that are carried through some selected sequences of labels.
This could be useful in a variety of applications:
in natural language processing patterns could correspond to certain syntactic constructions or stable idioms;
in protein secondary structure prediction --- to sequences of dihedral angles associated with stable configurations such as $\alpha$-helixes;
in gene prediction --- to sequences of nucleotides with supposed functional roles such as ``exon'' or ``intron'', specific codons, etc.

But along with long-range interactions, modeled by a set of words $\Lambda$, there could be interactions that have a very non-local nature.
Consider, for example, a language model problem, i.e. a problem of building a probabilistic model of sentences in a certain language.
Standard and the simplest way to build a language model is N-grams approach, which is equivalent to representing the probability of a sentence as a pattern-based CRF where the set $\Lambda$ is equal to the set of frequent N-grams.
It is a well-known fact that for most of natural languages, a set of frequent N-grams for small $N$ is not a large number, which justifies application of \PB.
 But at the same time, sentences also have a syntactic structure
that sometimes can be described by a context-free grammar.
Such syntactic correlations could have a non-local structure, and cannot be encoded in the \PB framework.
Thus, we have a problem of introducing grammar-based structures into a model. One of approaches to implement this idea can be found, e.g., in \cite{BenediS00}.

\begin{figure}
\centering
\includegraphics[width=0.2\textwidth]%
{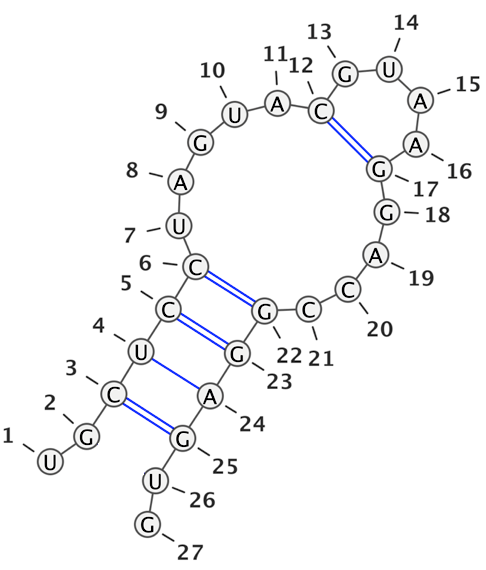}
\caption[my caption]{An RNA secondary structure without pseudoknots. Every pair of complementary nucleotides that have hydrogen bonds with each other correspond to opening and closing brackets in a correct bracket structure:

{
\def\E#1{$\underset{.}{\mbox{\tt #1}}$}
\def\L#1{$\underset{(}{\mbox{\tt #1}}$}
\def\R#1{$\underset{)}{\mbox{\tt #1}}$}

~~\E{U}\E{G}\L{C}\L{U}\L{C}\L{C}\E{U}\E{A}\E{G}\E{U}\E{A}\L{C}\E{G}\E{U}\E{A}\E{A}\R{G}\E{G}\E{A}\E{C}\E{C}\R{G}\R{G}\R{A}\R{G}\E{U}\E{G}
}


Equivalently, the sequence can be parsed according to a context-free grammar defined by the set of nonterminals $\{S\}$, the set of terminals $\{{\tt G},{\tt A},{\tt U},{\tt C}\}$,
and rules $S\rightarrow SS$, $S\rightarrow {\tt G}|{\tt A}|{\tt U}|{\tt C}$, $S\rightarrow {\tt G} S {\tt C}|{\tt C}S{\tt G}|{\tt U}S{\tt A}|{\tt A}S{\tt U}$.
}
\label{fig:ribosome}
\end{figure}

Let us give another example of a problem that has a similar flavour. Suppose that our task is to identify, for a given RNA sequence $z\in\{{\tt G},{\tt A},{\tt U},{\tt C}\}^n$,
certain properties of each nucleotide $z_i$ (e.g.\ discretized dihedral angles).
Suppose we have 4 vocabularies of such properties $D_{\tt G},D_{\tt A},D_{\tt U},D_{\tt C}$, one for each nucleotide. Thus, we have again the sequence labeling problem where labels alphabet is $D=\cup_{t}D_t$.
It is natural to model local interactions by pattern-based models, if ``sparsity'' assumption is satisfied. But RNA's secondary structure can also be modeled through context-free grammars
(see Fig. \ref{fig:ribosome}). Now we can introduce a context-free grammar with nonterminals set $\{S\}$, terminals set $V$, and weighted rules
$S\rightarrow SS$, $S\rightarrow a, a\in V$, $S\rightarrow x S y$, where $x\in D_t$, $y\in D_{t'}$ and $t,t'$ are two complementary nucleotides. This is a generalization of the context-free grammar introduced in Fig. \ref{fig:ribosome}. Then any potential labeling for an input RNA sequence could be optimally parsed according to this grammar, and this parsing can be interpreted as a system of interacting nucleotides in the sequence, i.e. secondary structure. Moreover, the weight of the parsing is a sum of binary terms defined on labels where each term is an interaction potential between corresponding nucleotides. This weight could also be included into a probabilistic model.
Note that in the last model an interaction that is associated with weighted rule $S\rightarrow x S y$ cannot be modeled by pattern-based CRFs.

These examples motivate the following model.
Consider a weighted context-free
grammar $\Gamma = \left(D,N,S,R,\nu\right)$, where $N$ is the alphabet of nonterminals, $S\in N$ is the initial symbol, $R$ is a set of rules, $\nu:R\rightarrow \mathbb{R}$ is a weighting function (that  could depend on parameters $\theta$ and observation $z$). A {\em Grammatical Pattern-Based model }(\GPB) is defined by probability distribution $p_\theta(x,\lambda|z)\sim \exp\{-E_\theta(x,\lambda|z)\}$ with the energy
\begin{equation}
E_\theta(x,\lambda|z)=E^{pb}_\theta(x|z)+C_{\Gamma (\theta)}(x,\lambda|z)
\label{eq:Grammatical0}
\end{equation}
where $C_{\Gamma (\theta)}(x,\lambda|z)$ is the cost of derivation (parsing) $\lambda$ of $x$ according to $\Gamma (\theta)$.
We view this as a rather natural way to combine \PB and \WCFG:
defining energy as a sum of terms that encode different constraints has a long history in the CRF literature.

\noindent {\bf Contributions.}
This paper investigates the complexity of several inference tasks in a \GPB.
Our focus is on the problem of computing a Maximum a Posteriori (MAP) solution $(x,\lambda)$,
i.e.\ minimizing energy \eqref{eq:Grammatical0}. We show that this can be done in polynomial
time; complexities are stated in the end of this section.
We also discuss how our algorithms can be adapted to compute in polynomial time sums $\sum_{x,\lambda}\exp\{-E_\theta(x,\lambda|z)\}$
and $\sum_{\lambda}\exp \{-E_\theta(x,\lambda|z)\}$ (for  given $x,z$).

\noindent {\bf Related work.}
Pattern-based CRFs with key inference
algorithms first appeared in \cite{Ye:NIPS09}.
Refined versions of these algorithms and an efficient sampling technique were described in \cite{TK:ICML}.
Applications of \PB considered in the literature so far include handwritten character recognition, identification of named entities from text  \cite{Ye:NIPS09}, optical character recognition \cite{Qian:ICML09} and the protein dihedral angles prediction problem \cite{TK:ICML}.
Further generalization of the model proposed in \cite{Qian:ICML09} considered a pattern as a set of strings rather than one single string. Another direction of generalization \cite{Nguyen:11} extends the segments of variables on which patterns are defined by allowing correspondence of each label of a pattern to a successive repetition of it on a line.

Probably the closest to ours is a line of research represented by the work \cite{BenediS00}.
They proposed a probabilistic {\em Hybrid model} that also integrates local correlations (namely, the $N$-gram model which is
a special case of \PB) and stochastic grammars.
Unlike us, they define the probability successively (their model is slightly different, but is equivalent to the following):
\begin{equation*}
\begin{array}{lr}
p(x_1\ldots x_n)=\prod_{i=1}^n p(x_i|x_1\ldots x_{i-1}) \\
p(x_i|x_1\ldots x_{i-1}) =\alpha p^{\tt pb}\left(x_i|x_1\ldots x_{i-1}\right) +\\
\hspace{72pt}+(1-\alpha)p^{\Gamma}(x_i|x_1\ldots x_{i-1}) \hspace{240pt}\\
p^{\Gamma}(x_i|x_1\ldots x_{i-1}) = \frac{\sum_{y_{i+1:n},\lambda}p^{\Gamma}(x_1\ldots x_{i},y_{i+1:n},\lambda)}{\sum_{y_{i:n},\lambda}p^{\Gamma}(x_1\ldots x_{i-1},y_{i:n},\lambda)}
\end{array}
\end{equation*}
where $p^{\tt pb}$ is an $N$-gram (a pattern-based) term and $p^{\Gamma}$ is a grammar term (defined as $\sim e^{-C_{\Gamma}(x,\lambda)}$), and $\alpha\in [0,1]$ is some parameter that mix two models.
The Hybrid model have been applied to various problems,
from modeling text 
to RNA secondary structure prediction \cite{Salvador01rna}.

Let us compare computational requirements of the two models, assuming that there is no dependence on observation $z$
(as was the case in  \cite{BenediS00,Salvador01rna}).
Both \GPB and the Hybrid models allow efficient computation of probability $p(x_1\ldots x_n)$ for a given word
(for the former it equals $p(x)=\sum_\lambda p(x,\lambda)$).
As can be shown, for \GPB this can be done in $O(|R|n^3)$ time.\footnote{We assume that the normalization
constant $\sum_{x,\lambda}\exp\{-E_\theta(x,\lambda)\}$ has been precomputed; this can be done in polynomial
time. Alternatively, this constant can be ignored if we are interested only in probability ratios.}
For the Hybrid model we would need $O(|R|n^3)$ for evaluating each term $p(x_i|x_1\ldots x_{i-1})$,
resulting in $O(|R|n^4)$ total time.

Another important task is computing labeling $x$ with the highest probability.
A tricky definition of total probability in the Hybrid makes this a non-trivial problem
(probably intractable), though sampling is easy.
We conjecture that in \GPB maximizing $p(x)$ over $x$ is also intractable;
however, we can do efficiently the next best thing, namely maximize $p(x,\lambda)$ over $x,\lambda$.

Therefore, we argue that  \GPB  has computational advantages over the Hybrid model.
Furthermore, a \GPB model can be trained using standard techniques such as the maximum likelihood principle
 (with a gradient-based on an EM-like method) or the struct-SVM approach with hidden variables.
The former requires computing sums of the form $\sum_{x,\lambda}\exp\{-E_\theta(x,\lambda)\}$
while the latter uses minimization of $E_\theta(x,\lambda)$ as a subroutine;
both tasks can be performed in polynomial time.

Another way of mixing $N$-grams correlations with context-free grammars was given in \cite{Wang00aunified}. Their model
is quite different from ours, and uses an unweighted version of \CFG.

Finally, we would like to mention that the idea of combining a certain sequential model with context-free grammars can be traced back to~\cite{Bar-Hillel}.
One of their classical results is as follows: if $\Gamma$ is a context-free grammar given in Chomsky Normal Form ($\CNF$) with $m$ nonterminals
and $r$ rules and $F$ is a finite-state automaton with $s$ states, then intersection of languages defined by them is a
context-free language that can be described by a context-free grammar with $ms^2$ nonterminals and $rs^3$ rules.
With some work, this result together with a version  of the \CYK parsing algorithm from~\cite{Katsirelos} could give
an alternative way to derive the algorithm for general grammars.
Instead of following this route, we chose to present the algorithm for general grammars directly.

\noindent{\bf Patterns interaction grammar.}
The algorithm for general \WCFGs has a rather high complexity (although polynomial).
An interesting question is whether there are special cases that admit faster inference.
We identified one such case that 
we call  {\em interaction grammars}; it is given by nonterminals set
$N=\left\{S\right\}$ and a set of rules:
\begin{equation}
R = \left\{ {\begin{array}{lr}
  {S \to SS; S \to a \in D\cup\{\varepsilon\}} \\
  {S \to u S v ,\left( {u ,v } \right) \in P}
\end{array}} \right\}
\end{equation}
where $P$ is a certain subset of $\Lambda\times\Lambda$. We will denote such grammar as $\Gamma (P)$.
We also restrict that only rules of the third type can have nonzero weights.
Note that the grammar that we described in the second motivating problem example, i.e. RNA sequence labeling, is of this kind.

\begin{figure}[t]
\vskip 0.05in
\small
\begin{center}
\includegraphics[scale=0.38]{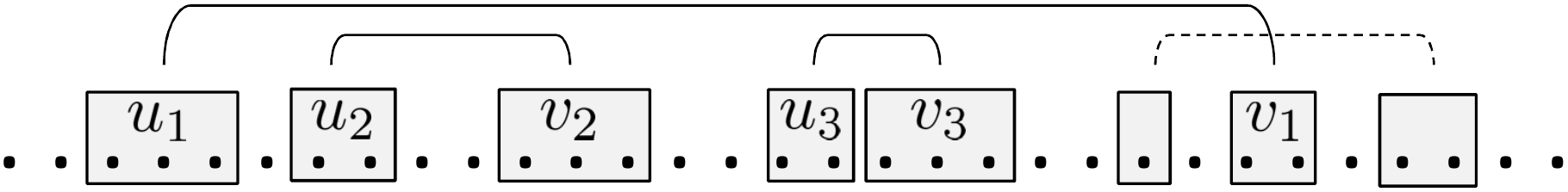} \vspace{-11pt} \\
\end{center}
\caption{Interactions that can be modeled form a bracket structure. The interaction shown by a dashed line
cannot be counted simultaneously with the interaction $(u_1,v_1)$.
}
\label{fig:brackets}
\vskip -0.1in
\end{figure}

Interaction grammars strengthen the \PB model by allowing non-local interactions between patterns,
albeit with some limitations on such interactions. Roughly speaking, two interactions $(u,v)$ and $(u',v')$
can be counted simultaneously only if they are either nested or do not overlap (Fig.~\ref{fig:brackets}).

For computational reasons we will also consider a further restriction in which the depth of inclusion does not exceed  some constant $d$.
(In the example in Fig.~\ref{fig:brackets} this depth equals 2.) Such restriction can be expressed by the grammar with non-terminals
$N=\left\{S^0, \ldots, S^{d-1}, S=S^d\right\}$ and with the following set of rules:
\begin{equation}
R = \left\{ {\begin{array}{lr}
  {S^k \to S^{k}S^{k}, k=0,\ldots,d{\text{                     }}} \\
  {S^k \to S^{k-1}; S^0 \to a \in D\cup\{\varepsilon\}}  \\
  {S^k \to u S^{k-1}v , k=1,\ldots,d,\left( {u ,v } \right) \in P}
\end{array}} \right\}
\end{equation}
Again, only rules of the fourth type can have nonzero weights.
We call it an interaction grammar of depth $d$.

\noindent{\bf Minimization.} This paper focuses on minimization algorithms for energies of the form~\eqref{eq:Grammatical0} over $\lambda, x$. Without loss of generality we can eliminate $\lambda$ by defining new functional that depends only on $x$:
\begin{equation}
E_w(x|z)=E^{pb}_\theta(x|z)+C_{\Gamma}(x|z)
\label{eq:Grammatical}
\end{equation}
where $C_{\Gamma}(x|z)$ is the cost of a least-weight derivation of $x$ according to $\Gamma$.

We will consider the following three cases:
(i) general \WCFG;
(ii) interaction grammar of depth $d\ge 2$;
(iii) interaction grammar of depth 1. For all three cases we will present algorithms.
The complexity of solving these tasks is discussed below.
We denote $L=\sum_{\alpha\in\Gamma}|\alpha|$ to be total length of patterns
and $\ell_{\max}=\max_{\alpha\in\Gamma}|\alpha|$ to be the maximum length of a pattern.

For the most general case (i) we present $\Theta\left(|R|(nL)^3\right)$ algorithm that uses $\Theta\left(|N|(nL)^2\right)$ space if grammar is given in $\CNF$. This algorithm is based on dynamic programming and uses a very similar data structures (so called {\em messages}) to standard \CYK parsing
algorithm~\cite{Aho}. Thus, a standard way of obtaining $\CNF$ leads us to $\Theta\left((|P|+|D|)(nL)^3\right)$ algorithm for interaction grammars and
$\Theta\left(d(|P|+|D|)(nL)^3\right)$ for interaction grammars of depth $d$. Note that at the core of the algorithm we compute multiplication of two $nL\times nL$ matrices over a
semiring $\left(S, \oplus, \otimes \right) = \left(\mathbb{R}, \max, + \right)$ (so called {\em min-plus product}, or {\em the distance product}) which makes it cubic with respect to $nL$. It is well-known
 \cite{Zwick} that the distance product of matrices is computationally equivalent to {\em all-pairs shortest path problem}, and this leads us to more efficient algorithms for special cases of \WCFGs,
namely interaction grammars of depth $d$.

In this case we compute
a similar set of messages that are defined on $d+1$ levels (from 0 to $d$). There are two types of operations that we apply to messages at $d$ iterations: first we compute messages of level $i$ based on already computed messages of level $i-1$ (we call this {\em vertical messages passing}), second we solve the all-pairs shortest path problem on a certain graph.
In the worst case, total complexity of this algorithm for interaction grammars of depth $d$ does not differ from complexity of the general algorithm.
However, in the best case it can be much faster, namely $O\left((nL)^2(d|P|+d\log nL+|D|)\right)$.
Computational results on some synthetic data are given 
\ifTR
in section~\ref{sec:experiments}.
\fi
Moreover, when $d=1$, the complexity is always $O (|P| n L (l_{\min} \min(|D|,\log l_{\min})+|P|))$ where $l_{\min} = \min_{w\in\Lambda}|w|$.

\section{Notation and preliminaries}
First, we introduce a few definitions.
\begin{list}{$\bullet$}{\leftmargin=1em \itemindent=0em \itemsep=-2pt}
\item A {\em pattern} is a pair $\alpha=([i,j],w)$ where $[i,j]$ is an interval in $[1,n]$ and $w=w_i\ldots w_j$ is
a sequence over alphabet $D$ indexed by integers in $[i,j]$ ($j\ge i-1$).
The {\em length} of $\alpha$ is denoted as   $|\alpha|=|w|= j-i+1$. For pattern $\alpha=([i,j],w)$ we will also denote $i_{\alpha} = i$, $j_{\alpha} = j$, $w_{\alpha} = w$.
\item Symbols ``$\ast$'' denotes an arbitrary word or pattern
(possibly the empty word $\varepsilon$ or the empty pattern $\varepsilon_s\triangleq([s+1,s],\varepsilon)$ at position $s$).
The exact meaning will always be clear from the context.
Similarly, ``+'' denotes an arbitrary non-empty word or pattern.
\item The concatenation of patterns $\alpha=([i,j],v)$ and $\beta=([j+1,k],w)$ is the
pattern $\alpha\beta\triangleq([i,k],vw)$.
Whenever we write $\alpha\beta$ we assume that it is defined, i.e.\ $\alpha=([\cdot,j],\cdot)$ and $\beta=([j+1,\cdot],\cdot)$ for some $j$.
Also, if $u\in D^\ast$, then $\alpha u = ([i,j+|u|],vu)$ and $u\alpha = ([i-|u|,j],uv)$.
\item For a pattern $\alpha=([i,j],v)$ and interval $[k,\ell]\subseteq[i,j]$,
the {\em subpattern of $\alpha$ at position $[k,\ell]$} is the pattern $\alpha_{k:\ell}\triangleq([k,\ell],v_{k:\ell})$
where $v_{k:\ell}=v_k\ldots v_\ell$. \\
If $k=i$ then $\alpha_{k:\ell}$ is called a {\em prefix} of $\alpha$.
If $\ell=j$ then $\alpha_{k:\ell}$ is a {\em suffix} of $\alpha$.
\item If $\beta$ is a subpattern of $\alpha$, i.e.\ $\beta=\alpha_{k:\ell}$ for some $[k,\ell]$,
then we say that $\beta$ is {\em contained} in $\alpha$. 
This is equivalent to the condition $\alpha=\ast\beta\ast$.
\item $D^{i:j}=\{([i,j],v)\:|\:v\in D^{[i,j]}\}$ is the
set of patterns with interval $[i,j]$.
\item
For a pattern $\alpha$ let $\alpha^-$ be the
prefix of $\alpha$ of length $|\alpha|-1$; if $\alpha$ is empty
then $\alpha^-$ is undefined.
\end{list}
We will consider the following  problem.
Let $\Pi^\circ$ be the set of patterns of words in $\Lambda$ placed at all possible positions:
$\Pi^\circ=\{([i,j],\alpha)\:|\:\alpha\in\Lambda)\}$.
Define the cost of pattern $x\in D^{i:j}$ via
\begin{equation}
\!\!\!\!\!\!f( x ) =\!\!\!\!\! \sum_{\alpha\in\Pi^\circ,x=\ast\alpha\ast} \!\!\! c_\alpha
\label{eq:F}
\end{equation}
where $c_\alpha\!\in\! R,\alpha\in\Pi^\circ$ are fixed constants.
Let $C_\Gamma(x)$ be the cost of a least-weight derivation of $x$ in $\Gamma$.
Our goal is to compute
\begin{equation}
M=\min_{x\in D^{1:n}} F(x)\label{eq:Mmin}
\end{equation}
where $F(x) = f(x)+C_{\Gamma}\left(x\right)$.

We select set $\Pi$ as
the set of prefixes of patterns in $\Pi^\circ$:
\begin{equation}
\Pi=\{\alpha\:|\:\exists\alpha\ast\in\Pi^\circ\}
\label{eq:GLADGAKGADF}
\end{equation}

\begin{list}{$\bullet$}{\leftmargin=1em \itemindent=0em \itemsep=-2pt}
\item For an index $s\in[0,n]$ we denote $\Pi_s$ to be the set of patterns in $\Pi$
that end at position $s$: $\Pi_s=\left\{\alpha\in\Pi | j_{\alpha} = s\right\}$.
\item For an arbitrary pattern $\alpha$, $lsp(\alpha)$ denotes the longest suffix of $\alpha$ that belongs to $\Pi_{j_{\alpha}}$.
\end{list}

\myparagraph{Graph $G[\Pi_s]$.} The following construction will be used throughout the paper. Given a set of patterns $\Pi$
and index $s$, we define $G[\Pi_s]=(\Pi_s,E[\Pi_s])$ to
be a directed graph with the following set of edges: $(\alpha,\beta)$ belongs to $E[\Pi_s]$ for $\alpha,\beta\in \Pi_s$ if $\alpha$ is a proper suffix of $\beta$ ($\beta=+\alpha$)
and $\Pi_s$ does not have an ``intermediate'' suffix $\gamma$ of $\beta$ with $|\beta|>|\gamma|>|\alpha|$. It can be checked that graph $G[\Pi_s]$ is a directed forest.
If $\varepsilon_s\in \Pi_s$ then $G[\Pi_s]$ is connected and therefore is a tree. In this case we treat $\varepsilon_s$ as the root.
An example is shown in Fig.~\ref{fig:graph}.

\begin{figure}[t]
\vskip 0.0in
\small
\begin{center}
\begin{tabular}{c@{\hspace{2pt}}c@{\hspace{1pt}}c@{\hspace{0pt}}c@{\hspace{0pt}}c@{\hspace{0pt}}c}
.&.&.&.&.&. \\
.&.&.&.&.&0 \\
.&.&.&.&.&1 \\
.&.&.&.&1&0 \\
.&.&.&1&0&0 \\
.&.&.&1&0&1 \\
.&.&1&0&0&0 \\
.&.&1&0&1&0
\end{tabular}
\normalsize
\hspace{0pt}
\begin{tabular}{c}
\includegraphics[scale=0.20]{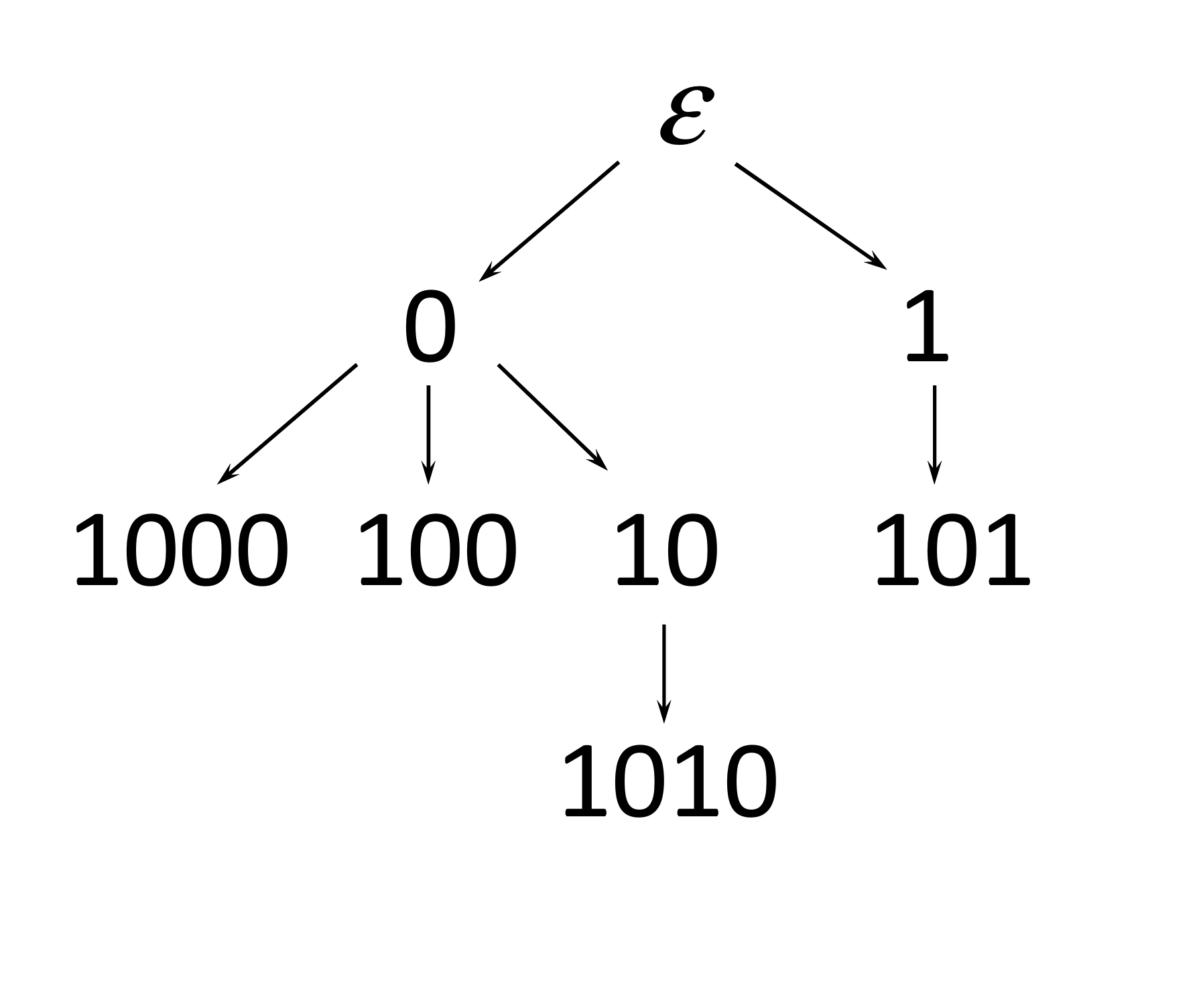} \vspace{-11pt} \\
\end{tabular}
\end{center}
\caption{Graph $G[\Pi_s]$ for the set of 8 patterns shown on the left (for brevity, their intervals are not shown;
they all end at the same position $s$.) This set of patterns would arise if $\Gamma=\{0,1,1000,1010\}$ and
$\Pi$ was defined as the set of all prefixes of patterns in $\Pi^\circ$.
}
\label{fig:graph}
\vskip -0.1in
\end{figure}

\myparagraph{Computing partial costs.} Recall that $f(\alpha)$ for a pattern $\alpha$
is the cost of all patterns inside $\alpha$ (eq.~\eqref{eq:F}).
We also define $\phi(\alpha)$ to be the cost of only those patterns that are suffixes of $\alpha$:
\begin{equation}
\!\!\!\!\!\!\phi( \alpha ) =\!\!\!\!\! \sum_{\beta\in\Pi^\circ,\alpha=\ast\beta} \!\!\! c_\beta
\end{equation}
In algorithms below we will use the following quantities: $\phi(\alpha), f(\alpha)$ for $\alpha\in\Pi$ and $f(\alpha\beta)$ for $\alpha,\beta\in\Pi$. Let us show
how to compute them efficiently.
\begin{lemma}
Values $\phi(\alpha),f(\alpha)$ for all $\alpha\in\Pi$
can be computed using $O(|\Pi|)$ additions and values $f(\alpha\beta), \alpha,\beta\in\Pi$
can be computed using $O(L|\Pi|)$ additions.
\label{lemma:varphis}
\end{lemma}
\ifTR
\begin{proof} To compute $\phi(\cdot)$ for patterns $\alpha\in\Pi_s$, we use the following procedure: (i) set $\phi(\varepsilon_s):=\mathds{1}$;
(ii) traverse edges $(\alpha,\beta)\in E[\Pi_s]$ of tree
$G[\Pi_s]$ (from the root to the leaves) and
set
$$
\phi(\beta):=\begin{cases}
\phi(\alpha) + c_\beta & \mbox{if }\beta\in\Pi^\circ \\
\phi(\alpha)  & \mbox{otherwise}
\end{cases}
$$
After computing $\phi(\cdot)$, we go through indexes $s\in[0,n]$ in increasing order and set
$$
f(\varepsilon_s):=0,\quad f(\alpha):=f(\alpha^-)+ \phi(\alpha)\quad\forall \alpha\in \Pi_s-\{\varepsilon_s\}
$$
$$
f(\alpha\beta):=f(\alpha\beta^-)+ \phi(lsp(\alpha\beta)),\forall \alpha\in \Pi_{i_{\beta}-1},\beta\in\Pi_s-\{\varepsilon_s\}
$$
\end{proof}
\fi


\section{Minimization for general grammars}\label{sec:ring}
In this section we describe our algorithm for minimizing energy \eqref{eq:Mmin} for general grammars.
The idea is to compute certain ``messages'' by a dynamic programming procedure. Note that the same system of messages
 will be computed by algorithms for interaction grammars, though the computation strategies will differ.

We will assume that a \WCFG $\Gamma = \left(D,N,S,R,\nu\right)$ is given in $\CNF$, i.e. all of its rules belong to one of the following types:
\[
\begin{array}{*{20}{c}}
 A\rightarrow BC,\,\,\, A,B,C\in N\\
A \rightarrow w,\,\,\, A\in N, w\in D^*\\
S \rightarrow \varepsilon \hspace{26 mm}\\
\end{array}
\]
We will also assume that for all rules $A \rightarrow w\in R$, the word $w$ belongs to $\Lambda$.
The set $\Pi$ has a natural partial order $\succ$ on it, such that for $\alpha = \left([i_{\alpha},j_{\alpha}],w_{\alpha}\right)$ and $\beta = \left([i_{\beta},j_{\beta}],w_{\beta}\right)$:
\begin{equation}
 \beta \succ \alpha \quad \Leftrightarrow\quad j_{\beta} > j_{\alpha}, i_{\beta}\geq i_{\alpha}, \alpha_{i_{\beta}:j_{\alpha}}=\beta_{i_{\beta}:j_{\alpha}}
\label{eq:order}
\end{equation}
In the previous definition if $i_{\beta}>j_{\alpha}$, then the condition $ \alpha_{i_{\beta}:j_{\alpha}}=\beta_{i_{\beta}:j_{\alpha}}$ is satisfied by definition.

We will compute messages $M_{A}\left(\alpha,\beta\right)$ for $A\in N$ and any pair of patterns $\alpha,\beta\in \Pi$ such that $ \beta \succ \alpha$.
Message $M_{A}\left(\alpha,\beta\right)$ has the following interpretation. We consider interval of variables $x_{i_{\alpha}:j_{\beta}}$ and all their assignments such that
$x_{i_{\alpha}:j_{\alpha}}=w_{\alpha}$, $x_{i_{\beta}:j_{\beta}}=w_{\beta}$ and for any $(\beta,\gamma)\in E[\Pi_{j_{\beta}}]$, $x_{i_{\gamma}:j_{\gamma}}\ne w_{\gamma}$.
On these variables we form a new functional that consists of those patterns (of \PB part) of \eqref{eq:Mmin} that belong to this interval plus $C_{\Gamma_{A}}(x_{j_{\alpha}+1:j_{\beta}})$ where $\Gamma_{A}$ is the same as $\Gamma$ but with the initial symbol changed to $A$. A minimum of this functional over defined set of assignments is our message $M_{A}\left(\alpha,\beta\right)$. I.e.:
\begin{equation}
M_{A}\left(\alpha,\beta\right) = \min_{\substack{x: x=\alpha\ast = \ast\beta \\ \forall (\beta,\gamma)\in E[\Pi_{j_{\beta}}]: \\ x\ne \ast\gamma}} f\left(x\right) + C_{\Gamma_{A}}(x_{j_{\alpha}+1:j_{\beta}})
\label{eq:Message}
\end{equation}

\begin{figure}[t]
\vskip 0.05in
\small
\begin{center}
\includegraphics[scale=0.38]{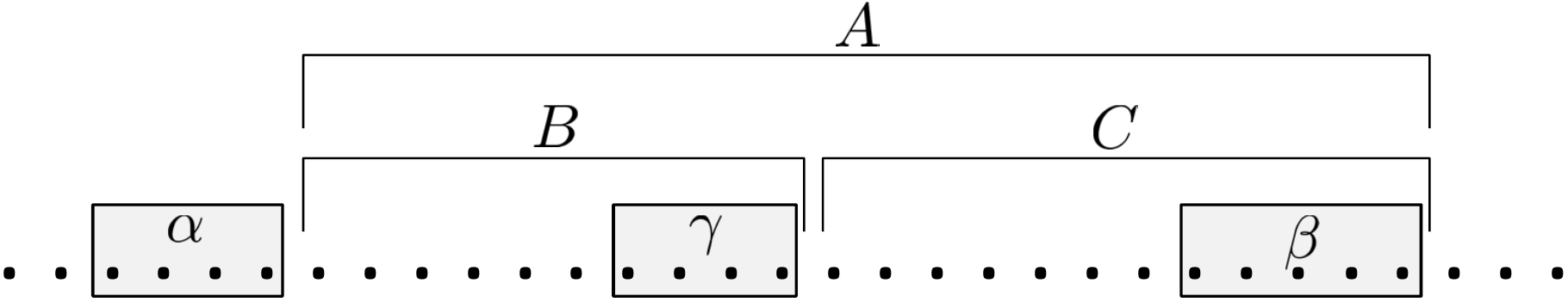} \vspace{-11pt} \\
\end{center}
\caption{Computing message $M_A(\alpha,\beta)$ assuming that $A\rightarrow BC$ is the next rule to be applied
and $\gamma$ is the longest suffix for $B$ which is in $\Pi$.
}
\label{fig:ABC}
\vskip -0.1in
\end{figure}

As in the \CYK algorithm, the main idea of our approach is to compute message $M_{A}\left(\alpha,\beta\right)$ based on the assumption that the best parsing of $x_{j_{\alpha}+1:j_{\beta}}$
for optimal $x$ starts by applying the rule $ A\rightarrow BC$.
Therefore, we need to consider possible divisions of $[j_{\alpha}+1,j_{\beta}]$ into two parts corresponding to nonterminals $B$ and $C$.
To count patterns that cross the boundary between $B$ and $C$, we also need to know the pattern $\gamma$ with which the first part ends (Fig.~\ref{fig:ABC});
such ``crossing'' patterns will be included in $M_C(\gamma,\beta)$. 
The resulting procedure is given in Algorithm~\ref{alg:Gen}.
Note that in~\eqref{eq:alg:update} we subtract $f(\gamma)$ to avoid counting patterns inside $\gamma$ twice.
\begin{algorithm}
\caption{Computing minimum of $F(x)$}\label{alg:Gen}
\begin{algorithmic}[1]
\STATE
set $M_A(\alpha,\beta):=+\infty$ for all  messages
\STATE {\bf for} each rule $A \rightarrow w \in R$ and each $\alpha\in\Pi$ set
\vspace{-3pt}
\begin{equation}
\label{eq:Message1}
M_A(\alpha,lsp\left(\alpha w\right)):=f\left(\alpha w\right)+\nu\left(A \rightarrow w\right)
\end{equation}
\vspace{-15pt}
\STATE {\bf for} each $\alpha,\beta\in\Pi:\beta\succ\alpha$ in the order of increasing $j_{\beta}-j_{\alpha}$  set $M_{A}\left(\alpha,\beta\right):=$
\begin{equation}
\label{eq:alg:update}
\hspace{-19pt}
\begin{array}{*{20}{c}}
\hspace{-5pt}\min\limits_{r=A\rightarrow BC\in R}\min\limits_{\substack{\gamma\in\Pi:\\ \beta\succ\gamma\succ\alpha}}\hspace{-5pt}M_{B}\left(\alpha,\gamma\right) \hspace{-3pt}
+\hspace{-3pt} M_{C}\left(\gamma,\beta\right)\hspace{-3pt}- \hspace{-3pt}f\left(\gamma\right)\hspace{-3pt}+\hspace{-3pt}\nu\left(r\right)  \\
\end{array}
\end{equation}
\vspace{-10pt}
\STATE {\bf return}
$
M:=\min\limits_{\alpha\in \Pi_{n}} M_{S}(\varepsilon_0,\alpha)
$
\end{algorithmic}
\end{algorithm}

\begin{theorem} Algorithm~\ref{alg:Gen} is correct, i.e.\ it
returns the correct value of $M=\min_x F(x)$.
\label{th:correctnessGen}
\end{theorem}

\ifTR
\subsection{Proof of Theorem~\ref{th:correctnessGen}}
First, we will prove the following result.
\begin{lemma}\label{lemma:Gen}  For each $\alpha,\beta,\gamma\in\Pi:\beta\succ\gamma\succ\alpha$ and a rule $r=A\rightarrow BC\in R$:
\[
M_{A}\left(\alpha,\beta\right) \leq M_{B}\left(\alpha,\gamma\right)+M_{C}\left(\gamma,\beta\right) - f\left(\gamma\right)+\nu\left(r\right)
\]
\end{lemma}
\begin{proof}
Suppose that $x^1_{i_{\alpha}:j_{\gamma}}$ is an optimal argument for message $M_{B}\left(\alpha,\gamma\right)$ and
$x^2_{i_{\gamma}:j_{\beta}}$ is an optimal argument for message $M_{C}\left(\gamma,\beta\right)$. Since both these arguments are equal to $w_{\gamma}$ on their intersection, there is $x_{i_{\alpha}:j_{\beta}}$ such that $x^1_{i_{\alpha}:j_{\gamma}}=x_{i_{\alpha}:j_{\gamma}}$ and $x^2_{i_{\gamma}:j_{\beta}}=x_{i_{\gamma}:j_{\beta}}$. Then if we use  $x_{i_{\alpha}:j_{\beta}}$
as an argument for message $M_{A}\left(\alpha,\beta\right)$ and parse $x_{j_{\alpha}+1:j_{\beta}}$ by applying rule $A\rightarrow BC$ and then deriving
$x_{j_{\alpha}+1:j_{\gamma}}$ from $B$ and $x_{j_{\gamma}+1:j_{\beta}}$ from $C$  (optimally, i.e. with minimal weight), we obtain an upper bound for $M_{A}\left(\alpha,\beta\right)$ that
equals $M_{B}\left(\alpha,\gamma\right) + M_{C}\left(\gamma,\beta\right)-f\left(\gamma\right)+\nu\left(r\right)$.
Here we needed to subtract $f\left(\gamma\right)$ in order to avoid counting some patterns twice.
\end{proof}

Let us now prove  by induction that Algorithm~\ref{alg:Gen} correctly computes all messages.
It can be checked that in step 2, the statement that $\beta = lsp(\alpha w)$, is equivalent that $\beta\in\Pi$ is such that $\beta\succ\alpha$, $j_{\beta}-j_{\alpha}=|w|$, $\alpha w = \ast\beta$ and $\alpha w\ne \gamma$ for any $(\beta,\gamma)\in E[\Pi_{j_{\beta}}]$. Or, equivalently, $x=\alpha w$ can serve as an argument for message  $M_A(\alpha,\beta)$ (see~\eqref{eq:Message} for the definition of the set of arguments).

Then we simply compute messages $M_A(\alpha,\beta)$ based on the assumption that
$f(x)+C_{\Gamma_{A}}(x_{j_{\alpha}+1:j_{\beta}})$ (see message definition) attains its minimum when we set $x=\alpha w$  and
$x_{j_{\alpha}+1:j_{\beta}} = w$ is parsed via rule $A\rightarrow w$. In this case, $C_{\Gamma_{A}}(x_{j_{\alpha}+1:j_{\beta}}) = \nu\left(A \rightarrow w\right)$
and $f(x) = f\left(\alpha w\right)$ which explains the expression to be computed.

Those messages for which we computed their values correctly in steps 1-2 will serve as an induction base. Now let us consider a  message $M_A(\alpha,\beta)$
for which steps 1-2 compute non-optimal value. This means that an optimal argument $x_{j_{\alpha}+1:j_{\beta}}$ should be parsed by applying first some rule $A\rightarrow BC$. Suppose that in an optimal parsing nonterminal $B$ is responsible for a word  $x_{j_{\alpha}+1:t}$
and nonterminal $C$ for residual  $x_{t+1:j_{\beta}}$. Consider all patterns that are satisfied on optimal $x$ and contain variable $x_{t}$.
Suppose that $\gamma$ is one of them which has the leftmost $i_{\gamma}$. Denote $\gamma' = \gamma_{i_{\gamma}:t}$. Then $x_{i_{\alpha}:t}$ is also an optimal argument for message $M_{B}\left(\alpha,\gamma'\right)$ and $x_{t+1-|\gamma'|:j_{\beta}}$ is an optimal argument for $M_{C}\left(\gamma',\beta\right)$. It is easy to check that, otherwise, we could change $x$ in segment $x_{i_{\alpha}:t}$ (or in segment $x_{t+1-|\gamma'|:j_{\beta}}$) to a locally optimal that will improve an overall value. Therefore,
\[
M_{A}\left(\alpha,\beta\right) \hspace{-2pt}=\hspace{-2pt} M_{B}\left(\alpha,\gamma'\right)\hspace{-1pt}+\hspace{-1pt}M_{C}\left(\gamma',\beta\right)\hspace{-1pt} -\hspace{-1pt} f\left(\gamma'\right)\hspace{-1pt}+\hspace{-1pt}\nu\left(A\rightarrow BC\right)
\]
Together with Lemma~\ref{lemma:Gen} this implies that $M_{A}\left(\alpha,\beta\right)$ is equal to the right side of \eqref{eq:alg:update}.

Now we only need to notice that we reduced the computation of a message to other messages with smaller $j_{\beta}-j_{\alpha}$. Therefore, correctness of the computation depends on whether we calculate correctly messages $M_{A}\left(\alpha,\beta\right)$ with such small values of $j_{\beta}-j_{\alpha}$, that a part of their optimal argument $x_{j_{\alpha}+1:j_{\beta}}$ cannot be parsed by a rule of the form $A\rightarrow BC$ (i.e. cannot be split). But, since we already know that we computed such messages correctly, we conclude that all final message values are correct.


\subsection{Algorithm's complexity}
The complexity of step 1 is $O\left((nL)^2\right)$.
In step 2 we go through at most $nL|R|$  different pairs $\left(A\rightarrow w,\alpha\right)$: $|R|$ for $A\rightarrow w$, $nL$ for $\alpha$. Since $w_{\beta}$ is a function of $w_{\alpha}$ and $w$, it can be computed during preprocessing. It can be seen that with proper preprocessing of the set of pairs $\left\{\left(A\rightarrow w,w_{\alpha}\right)\right\}$ we can compute $M_A(\alpha,\beta)$ in step 2 in time $O(1)$. Given that we already computed $f(\alpha w)$ according to Lemma~\ref{lemma:varphis}, the complexity of the step 2 is bounded by $O\left(nL|R|\right)$.

Finally, the complexity of step 3 can be bounded by $O\left(|R| |\Pi|^3\right) = O(|R|\cdot (nL)^3)$ which clearly dominates two previous steps and the computation of values in Lemma~\ref{lemma:varphis}. Therefore, the algorithm's overall complexity is $O\left(|R|(nL)^3\right)$.

\fi


\section{Algorithm for interaction grammars}
We will now describe an algorithm for interaction grammars of depth $d$,  $\Gamma^d (P)$.
Recall that only rules of the form $S^k \rightarrow u S^{k-1} v, (u,v)\in P$ can have nonzero weight; we denote this weight as $\theta^k_{u,v}$.
Unlike the general grammar case, we will not convert the grammar into a $\CNF$, but will use the initial form of it. Our algorithm is based on computing the same set of
messages $M_{S^k}\left(\alpha,\beta\right), k=0,\ldots,d$, which we will now denote simpler as $M_{k}\left(\alpha,\beta\right)$.

We will successively compute messages $M_{k}\left(\alpha,\beta\right)$ for $k=0,1,\cdots,d$. At the $k$th iteration we first compute message
$M_{k}\left(\alpha,\beta\right)$ based on the assumption that optimal parsing of $x_{j_{\alpha}+1:j_{\beta}}$ starts by applying some rule $S^{k}\rightarrow uS^{k-1}v$
(step 3 of Algorithm~\ref{alg:Int} below).
This part of computation can be done relatively fast. Then we have to update already computed messages $M_{k}\left(\alpha,\beta\right)$ based on the assumption that optimal parsing starts by dividing the argument $x_{j_{\alpha}+1:j_{\beta}}$ into pieces (by applying the rule $S^k\rightarrow S^kS^k$) (step 4). As we will see, this is equivalent to computing shortest paths in a graph with vertex set $\Pi$. 

\begin{algorithm}
\caption{Computing minimum of $F(x)$}\label{alg:Int}
\begin{algorithmic}[1]
\STATE {\bf for} each $\alpha,\beta\in\Pi:\beta\succ\alpha$ compute $M_0\left(\alpha,\beta\right)$
\FOR{$k=1,\ldots,d$}
\STATE {\bf for} each $\alpha,\beta\in\Pi:\beta\succ\alpha$
\begin{equation}
\label{eq:Message11}
\begin{array}{*{20}{c}}
\hspace{-21pt}M_k\left(\alpha,\beta\right)\hspace{-3pt}:= \hspace{-10pt}\min\limits_{(u,v,\gamma)\in \Omega\left(\beta\right)}\hspace{-7pt}M_{k-1}\left(lsp\left(\alpha u\right),\gamma\right)+\theta^k_{u,v}+\\
+f\left(\alpha u\right)-f\left(lsp\left(\alpha u\right)\right) + f\left(\gamma v\right) - f\left(\gamma\right) \\
\end{array}
\end{equation}
where $\Omega\left(\beta\right) \hspace{-3pt}=\hspace{-3pt} \{(u,v,\gamma)|(u,v)\hspace{-3pt}\in\hspace{-3pt} P, x_{\beta}=\ast v; \gamma\hspace{-3pt}\in\hspace{-3pt} \Pi, \gamma \hspace{-3pt}=\hspace{-3pt} \ast\beta_{1:|\beta|-|v|}, \forall (\beta,\delta)\hspace{-3pt}\in\hspace{-3pt} E[\Pi_{j_{\beta}}]\,\,\gamma \ne \ast\delta_{1:|\delta|-|v|}\}$.
\STATE {\bf compute} all-pairs shortest paths (denoted $SP\left(\alpha,\beta\right)$) for oriented graph $G=\left(\Pi,\succ\right)$ with costs of edges $c(\alpha,\beta):=M_k\left(\alpha,\beta\right) - f\left(\beta\right)$
\STATE {\bf for} each $\alpha,\beta\in\Pi:\beta\succ\alpha$
\begin{equation}
M_k\left(\alpha,\beta\right):=  SP\left(\alpha,\beta\right) + f\left(\beta\right)
\label{eq:SP}
\end{equation}
\ENDFOR
\STATE {\bf return}
$
M:=\min\limits_{k=\overline{0,d}}\min\limits_{\alpha\in\Pi_{n}} M_{k}(\varepsilon_0,\alpha)
$
\end{algorithmic}
\end{algorithm}
\begin{theorem} Algorithm~\ref{alg:Int} is correct, i.e. it
returns the correct value of $M=\min_x F(x)$ for an interaction grammar of depth $d$.
\label{th:correctnessInt}
\end{theorem}
\ifTR
\subsection{Proof of Theorem~\ref{th:correctnessInt}}
\begin{lemma}\label{lemma:IntH}  Suppose that $\alpha,\beta\in\Pi:\beta\succ\alpha$ and optimal argument $x_{i_{\alpha}:j_{\beta}}$ for message $M_{k}\left(\alpha,\beta\right)$ are such that optimal parsing of $x_{j_{\alpha}+1:j_{\beta}}$ from nonterminal $S^k$ starts by applying the rule $S^k \rightarrow u S^{k-1} v$. Then,
\begin{equation}
\begin{array}{*{20}{c}}
M_k\left(\alpha,\beta\right)= \min\limits_{\gamma\in\Omega}M_{k-1}\left(lsp\left(\alpha u\right),\gamma\right) +\theta^k_{u,v}+\\+f\left(\alpha u\right)-f\left(lsp\left(\alpha u\right)\right)+f\left(\gamma v\right)- f\left(\gamma\right) \\
\end{array}
\end{equation}
where
\[
\Omega  = \left\{ {\begin{array}{*{20}{lr}}
  {\gamma  \in \Pi |\gamma  =  * {\beta _{1:|\beta | - |v|}},\forall (\beta ,\delta ) \in E[{\Pi _{{j_\beta }}}],} \\
  {\gamma  \ne  * {\delta _{1:|\delta | - |v|}}}
\end{array}} \right\}
\]
\end{lemma}

\begin{proof}
Since first rule in the parsing of $x_{j_{\alpha}+1:j_{\beta}}$ is $S^k \rightarrow u S^{k-1} v$, then $x_{j_{\alpha}+1:j_{\beta}} = \ast v$. Therefore,
$w_{\beta} = \ast v$ or $+ w_{\beta} =  v$. The second option does not hold because we require that $\forall (\beta,\delta)\in E[\Pi_{j_{\beta}}]\,\,
x_{i_{\alpha}:j_{\beta}}\ne \ast\delta$.

Suppose now that $\gamma$ is the longest pattern from $\Pi_{j_{\beta}-|v|}$ for which $x_{i_{\alpha}:j_{\beta}-|v|} = \ast\gamma$.
According to the definition, $M_k\left(\alpha,\beta\right)= f\left(x_{i_{\alpha}:j_{\beta}}\right)+C_{\Gamma_{S^k}}\left(x_{j_{\alpha}+1:j_{\beta}}\right)$. We know that $C_{\Gamma_{S^k}}\left(x_{j_{\alpha}+1:j_{\beta}}\right) = \theta^k_{u,v}+C_{\Gamma_{S^{k-1}}}\left(x_{j_{\alpha}+1+|u|:j_{\beta}-|v|}\right)$. Now we will prove that
\begin{equation}
\label{eq:ExcInc}
\begin{array}{*{20}{c}}
f\left(x_{i_{\alpha}:j_{\beta}}\right) = f\left(x_{j_{\alpha}+1+|u|-|lsp(\alpha u)|:j_{\beta}-|v|}\right) + f\left(\alpha u\right)+\\
+f\left(\gamma v\right) -f\left(lsp\left(\alpha u\right)\right)- f\left(\gamma\right)\\
\end{array}
\end{equation}
First check that if a pattern $\delta\in\Pi_0$ is contained in $x_{i_{\alpha}:j_{\beta}}$ and is not counted in the first term of~\eqref{eq:ExcInc}, then we have only 2 options: either $\delta$ is contained in  $1)$ $\alpha u$, or in $2)$ $\gamma v$. Indeed, $\delta$ is not inside $\left[j_{\alpha}+1+|u|-|lsp(\alpha u)|,j_{\beta}-|v|\right]$ only if either $a)$ $j_{\delta} < j_{\alpha}+1+|u|-|lsp(\alpha u)|$ (then, obviously, $\delta$ is in $\alpha u$), or $b)$ $i_{\delta} > j_{\beta}-|v|$ (then $\delta$ is in $\gamma v$), or $c)$ $j_{\alpha}+|u|-|lsp(\alpha  u)|,j_{\alpha}+1+|u|-|lsp(\alpha u)| \in \left[i_{\delta}, j_{\delta}\right]$, or $d)$ $j_{\beta}-|v|,j_{\beta}-|v|+1\in \left[i_{\delta}, j_{\delta}\right]$. In case $c)$ $\delta$ is not in $\alpha u$ only if $j_{\delta} > j_{\alpha}+|u|$ which implies that  $\ast\delta_{i_{\delta}:j_{\alpha}+|u|} =  \alpha u, \delta_{i_{\delta}:j_{\alpha}+|u|}\in\Pi$, and therefore $\ast\delta_{i_{\delta}:j_{\alpha}+|u|} =  lsp( \alpha u)$, which contradicts that $\delta$ contains $x_{j_{\alpha}+|u|-|lsp(\alpha u)|}$. In  case $d)$ $\delta_{i_{\delta}:j_{\beta}-|v|}\in\Pi$ and by definition of $\gamma$,  $\gamma=\ast\delta_{i_{\delta}:j_{\beta}-|v|}$, which implies that $\delta$ is in $\gamma v$.

Formula~\eqref{eq:ExcInc} is obtained from inclusion-exclusion principle. Indeed, patterns in $x_{i_{\alpha},j_{\beta}}$, as we proved, is a union of patterns in  $x_{j_{\alpha}+1+|u|-|lsp(\alpha u)|:j_{\beta}-|v|}$, $\alpha u$ and $\gamma v$. Moreover, intersection of a set of patterns in $x_{j_{\alpha}+1+|u|-|lsp(\alpha u)|:j_{\beta}-|v|}$ and a set of patterns in $\alpha u$ ($\gamma v$) is a set of patterns in $lsp(\alpha u)$ ($\gamma$). The last thing we need to check is that patterns that are in both $\alpha u$ and in $\gamma v$ are all in $x_{j_{\alpha}+1+|u|-|lsp(\alpha u)|:j_{\beta}-|v|}$. It is easy to see that this is equivalent to $[i_{\gamma},j_{\alpha}+|u|]\subseteq [j_{\alpha}+1+|u|-|lsp(\alpha u)|,j_{\beta}-|v|]$ (if the first interval is empty then statement is obvious).The fact that $j_{\alpha}+|u| < j_{\beta}-|v|$ is obvious. Since $\gamma_{i_{\gamma}:j_{\alpha}+|u|}\in\Pi$ and $\alpha u = \ast\gamma_{i_{\gamma}:j_{\alpha}+|u|}$, then by definition of $lsp(\alpha u)$ we conclude that $i_{\gamma}>j_{\alpha}+|u|-lsp(\alpha u)$. Equality~\eqref{eq:ExcInc} is proved.

Now we see that
\begin{equation}
\begin{array}{*{20}{lr}}
M_k\left(\alpha,\beta\right) =f\left(x_{j_{\alpha}+1+|u|-|lsp(\alpha u)|:j_{\beta}-|v|}\right) \\
~~ +C_{\Gamma_{S^{k-1}}}\left(x_{j_{\alpha}+1+|u|:j_{\beta}-|v|}\right) +\theta^k_{u,v}+f\left(\alpha u\right)+f\left(\gamma v\right) \\
~~ -f\left(lsp\left(\alpha u\right)\right)- f\left(\gamma\right)\\
\end{array}
\end{equation}
Note that $x_{j_{\alpha}+1+|u|-|lsp(\alpha u)|:j_{\beta}-|v|}$ can serve as an argument for message
$M_{k-1} (lsp(\alpha u),  \gamma)$. Indeed, by definition of $\gamma$, for any $(\gamma,\delta)\hspace{-3pt}\in\hspace{-3pt} \Pi_{j_{\gamma}},\hspace{-1pt}  x_{j_{\alpha}+1\hspace{-1pt}+\hspace{-1pt}|u|\hspace{-1pt}-\hspace{-1pt}|lsp(\alpha u)|:j_{\beta}\hspace{-1pt}-\hspace{-1pt}|v|}\ne\ast\delta$ and we can consider any parsing of $x_{j_{\alpha}+1+|u|:j_{\beta}-|v|}$ from $S^{k-1}$. It can be shown that $f\left(x_{j_{\alpha}+1+|u|-|lsp(\alpha u)|:j_{\beta}-|v|}\right)+ C_{\Gamma_{S^{k-1}}}\left(x_{j_{\alpha}+1+|u|:j_{\beta}-|v|}\right)$ is  equal to
$M_{k-1}\left(lsp(\alpha u), \gamma\right)$, since formula~\eqref{eq:ExcInc} holds for any variations of $x$ for which $\alpha,\beta,\gamma$ preserve their properties. Therefore we conclude that $M_k\left(\alpha,\beta\right)= M_{k-1}\left(lsp\left(\alpha u\right),\gamma\right) +\theta^k_{u,v}+f\left(\alpha u\right)-f\left(lsp\left(\alpha u\right)\right)+f\left(\gamma v\right)- f\left(\gamma\right)$. To obtain the final
formula we need to add minimum over all $\gamma$.
\end{proof}

\begin{lemma}\label{lemma:IntV}  Let $G=\left(\Pi,\succ\right)$ be a weighted oriented graph with costs of edges $c(\alpha,\beta)={\widetilde M_k}\left(\alpha,\beta\right) - f\left(\beta\right)$, where ${\widetilde M_k}\left(\alpha,\beta\right)$ are values of messages after step 3 of Algorithm~\ref{alg:Int}. Suppose that $\alpha,\beta\in\Pi:\beta\succ\alpha$. Then, $M_{k}\left(\alpha,\beta\right) = SP\left(\alpha,\beta\right)+f\left(\beta\right)$, where $ SP\left(\alpha,\beta\right)$ is the length of a shortest path from $\alpha$ to $\beta$ in graph $G$.
\end{lemma}
\begin{proof}
Suppose that optimal argument $x_{i_{\alpha}:j_{\beta}}$ for message $M_{k}\left(\alpha,\beta\right)$ is such that optimal parsing of $x_{j_{\alpha}+1:j_{\beta}}$ from nonterminal $S^k$ starts by applying the rule $S^k \rightarrow S^{k} S^{k}$, then on one of resulting $S^{k}$ we again apply the same rule and etc.: $S^k \rightarrow S^{k} S^{k}\rightarrow S^{k} S^{k} S^{k} \rightarrow \cdots \rightarrow \underbrace{S^{k}\cdots S^{k}}_{s\rm{\,\,times}}$, and after this each $S^k$ unfolds according to some $S^k \rightarrow u S^{k-1} v\in R$.

Suppose that $x_{j_{\alpha}+1:l_1}$ is parsed from the first $S^k$, $x_{l_1+1:l_2}$ is parsed from the second $S^k$ and etc., until $x_{l_{s-1}+1:l_s},l_s=j_{\beta},$ is parsed from the last $S^k$. Also, $\gamma_i$ is longest pattern in $\Pi_{l_i}$ for which $x_{i_{\alpha}:l_i}=\ast\gamma_i$. Clearly, $\gamma_s = \beta$. For completeness, $\gamma_0=\alpha$.
Then it is easy to check that $\gamma_i\succ \gamma_{i-1}$ and (we give it without proof)
\vspace{-7pt}
\begin{equation}
\label{eq:path}
M_{k}\left(\alpha,\beta\right)-f\left(\beta\right) =\sum\limits_{i=1}^s {\widetilde M_k}\left(\gamma_{i-1},\gamma_i\right) - f\left(\gamma_i\right)\\
\end{equation}
We subtract $f\left(\gamma_i\right)$ each time in order avoid counting some patterns twice (like it would be in ${\widetilde M_k}\left(\gamma_{i-1},\gamma_i\right)+{\widetilde M_k}\left(\gamma_{i},\gamma_{i+1}\right)$).

And visa versa, if we have a chain $\gamma_s\succ\gamma_{s-1}\succ\cdots\succ\gamma_0$, then it corresponds to some pattern $x_{i_{\gamma_0}:j_{\gamma_s}}$ that is defined by requirement that $x_{i_{\gamma_i}:j_{\gamma_{i+1}}}$ is an optimal argument for message ${\widetilde M_k}\left(\gamma_{i},\gamma_{i+1}\right)$. A parsing of $x_{j_{\gamma_0}+1:j_{\gamma_s}}$ consists of first applying  $s-1$ times the rule $S^k \rightarrow S^{k} S^{k}$, and deriving $x_{j_{\gamma_{i-1}}+1:j_{\gamma_{i}}}$ with minimal weight from $i$-th $S^k$. The cost of such parsing plus pattern-based part is equal to $f\left(\gamma_l\right)+\sum\limits_{i=1}^s {\widetilde M_k}\left(\gamma_{i-1},\gamma_i\right) - f\left(\gamma_i\right)$.

Now we see that an optimal argument with its parsing corresponds to a sum of the form~\eqref{eq:path} and each such sum corresponds to some argument and its parsing. Therefore, an optimal argument with its parsing for $M_{k}\left(\alpha,\beta\right)$ can be found by computing a shortest path from $\alpha$ to $\beta$ in $G$ and the value of message is defined by~\eqref{eq:SP}.
\end{proof}


Lemmas~\ref{lemma:IntH} and \ref{lemma:IntV} imply that in step 5 of Algorithm~\ref{alg:Int} we correctly
compute messages $M_{k}\left(\alpha,\beta\right)$. It remains to verify that the expression in step 7 gives indeed the value of the optimum $M$;
this fact follows directly from definitions.
\fi

\subsection{Analysis of complexity}
\ifTR
First let us estimate the complexity of step 1. 
A message $M_{0}\left(\alpha,\beta\right)$ does not include any grammatical part and can be computed with the same
techniques as in~\cite{TK:ICML}. 
There it was proved that all messages (in current notations) of the form $M_{0}\left(\varepsilon_{0},\beta\right)$ for all $\beta\in\Pi$ can be
computed in time $O\left(nL|D|\right)$. Now we only have to note that any $\alpha\in \Pi$ defines its own pattern-based potential on the interval of variables $[i_{\alpha},n]$ that includes all patterns from $\Pi^\circ$ that satisfy: a) a pattern interval is a subinterval of $[i_{\alpha},n]$;b) a pattern is consistent with $\alpha$ on the intersection with $[i_{\alpha},j_{\alpha}]$. We assume weights of patterns to be the same, except for $c_{\alpha} = -C + c^{old}_{\alpha}$ where $C$ is a large constant. For each such pattern-based potential $p_{\alpha}$ we can compute in $O\left(nL|D|\right)$ time all messages of the form $M^{p_{\alpha}}_0\left(\varepsilon_{i_{\alpha}-1},\beta\right)$ that will be equal to $M_{0}\left(\alpha,\beta\right)-C$. Therefore, all messages $M_{0}\left(\alpha,\beta\right)$ can be computed in  $O\left((nL)^2|D|\right)$ time.
\fi

\ifTR
Now let us turn to step 3 of the algorithm where we make vertical message passing. The complexity of this part is $O\left(dnL\sum\limits_{\beta\in\Pi}\left|\Omega\left(\beta\right)\right|\right)$: $d$ times in a loop, $O(nL)$ variants of choosing $\alpha$, and $\left|\Omega\left(\beta\right)\right|$ variants for choosing $(u,v,\gamma)$ given $\beta\in\Pi$.
\begin{lemma}
$\sum\limits_{\beta\in\Pi}\left|\Omega\left(\beta\right)\right| \leq |P| nL$.
\end{lemma}
\begin{proof}
Let us consider  set $X=\{(\beta,u,v,\gamma)|\beta\in\Pi,(u,v,\gamma)\in \Omega\left(\beta\right)\}$. The cardinality of this set exactly equals the expression to be bounded.
Now given $(u,v)\in P$ and $\gamma\in\Pi$ we will show that if $(\beta,u,v,\gamma)\in X$ then $\beta = lsp(\gamma v)$; this will imply the  lemma.

Using definition of $\Omega\left(\beta\right)$ we reformulate $X$ as $\{(\beta,u,v,\gamma)|\beta\in\Pi,(u,v)\in P, x_{\beta}=\ast v; \gamma\in \Pi, \gamma = \ast\beta_{1:|\beta|-|v|}, \forall (\beta,\delta)\in E[\Pi_{j_{\beta}}]\,\,\gamma \ne \ast\delta_{1:|\delta|-|v|}\}$. Then with fixed $(u,v)\in P$ and $\gamma\in\Pi$ the definition requires that $\beta$ satisfies:
$\beta\in\Pi, x_{\beta}=\ast v, \gamma v = \ast\beta$ and $\forall (\beta,\delta)\in E[\Pi_{j_{\beta}}]\,\,\gamma v \ne \ast\delta$. It can be checked that this is equivalent to $ \beta = lsp(\gamma v)$.
\end{proof}

From this lemma we conclude that complexity of step 3 is $O\left(d(nL)^2|P|\right)$.
\fi

Step 4 is the most expensive step of the algorithm. To estimate its complexity,
we need to specify how we compute shortest paths.
\begin{theorem}
\label{additions}
$M=\min_x F(x)$ can be computed in time $O((nL)^2(d|P|+d\log nL+|D|) + nL\sum_{i=1}^{d}|E_i|)$, where
$E_i$ is the set of pairs $(\alpha,\beta):\beta\succ \alpha$ for which message $M_i(\alpha,\beta)$ was correctly computed in step 3 of Algorithm~\ref{alg:Int}.
\end{theorem}
\ifTR
\begin{proof}
First let us show that
there is a simple scaling procedure~\cite{Goldberg93} that makes all edge weights nonnegative. Let us add a new vertex $s$ (source) to $G=\left(\Pi,\succ\right)$ and connect this source with all vertices that have no incoming edges. We assign zero weights to new added edges. Since our graph  is acyclic, then single source shortest path algorithm will take only $O\left((nL)^2\right)$ time~\cite{Cormen}. This way we can compute values $r(\alpha),\alpha\in\Pi$ that are equal to the distance from $s$ to $\alpha$ in the new graph. Now it can be checked that
\begin{equation}
r(\alpha)-r(\beta)+c(\alpha,\beta)\geq 0, \beta\succ\alpha
\end{equation}
Therefore, we can define new weights by the following rescaling formula: $c'(\alpha,\beta) = r(\alpha)-r(\beta)+c(\alpha,\beta)$ and for every path from $\alpha$ to $\beta$ in the graph its new length $l'$ will depend on the old one $l$ by formula: $l' = r(\alpha)-r(\beta)+l$. Therefore, rescaled graph will have the same shortest paths as the initial one.

Now for this rescaled version we can apply standard algorithms for all-pairs shortest path problem with nonnegative edge weights.
We suggest an algorithm of Karger-Koller-Phillips~\cite{Karger+al:93} (or, of~Demetrescu-Italiano\cite{Demetrescu}). This algorithm has complexity $O\left(|E^*||V|+|V|^2\log |V|\right)$, where $V$ is a set of vertices, $E$ is a set of edges, and $E^*$ is the set of edges $(u,v)\in E$ for which one of optimal paths from $u$ to $v$ is equal to $\left\{(u,v)\right\}$.
Under a non-critical assumption we have $E^*=E_i$ for the graph $(V,E)$ in step 4 of the $i$-th iteration of the algorithm. \footnote{To achieve this equality, we can
subtract a sufficiently small value $\varepsilon > 0$ from initial edge weights; shortest paths of the graph with new weights will also be shortest for the graph with old weights.}
%
Therefore, the total complexity of Algorithm~\ref{alg:Int} will be $O((nL)^2(d|P|+d\log nL+|D|)+ nL\sum_{i=1}^{d}|E_i|)$.
\end{proof}

\fi
The complexity is usually dominated by the term $nL\sum_{i=1}^{d}|E_i|$:
in the worst case we have $|E_i| = O((nL)^2)$, and thus
the overall worst-case complexity is cubic with respect to $nL$ (as
of the algorithm for general context-free grammars).
In practice, however, $|E_i|$ can be smaller (reaching $|E_i|=O(nL)$ in
the best case),
and so the empirical runtime of the algorithm can be better than $O((nL)^3)$.
Experiments on a synthetic data given 
\ifTR
in  section~\ref{sec:experiments}
\fi
confirm this.

\subsection{Linear time algorithm for $d=1$}
In this section we focus on the case $d=1$.
First, we show that just a simplification of Algorithm~\ref{alg:Int} gives the following.
\begin{theorem} For $d=1$, $M=\min_x F(x)$ can be computed in time $O((nL)^2(|P|+|D|))$.
\end{theorem}
\ifTR
\begin{proof}
The specificity of this case is determined by the fact that we compute shortest paths only for one graph and in the end of the algorithm in step 7 we  need only shortest paths that starts from vertex $\varepsilon_{0}$. Therefore, instead of using all-pairs shortest path algorithm we can use only single source version of it.
Since there is $O\left(|V|+|E|\right)$ single source shortest path algorithm for general acyclic graphs~\cite{Cormen}, overall complexity of Algorithm~\ref{alg:Int} is $O\left((nL)^2(|P|+|D|)\right)$.
\end{proof}
\fi

In the previous algorithm we did not change data structures, though made computations with them more efficient. This way we cannot radically improve complexity, since only reading the input would take time quadratic with respect to $nL$.

We now describe an alternative approach which scales linearly with  $nL$.
The idea is to compute another set of messages
$S\left(g,\beta\right)$, where $\beta\in\Pi$ and $g$ is a {\em  state}
described by a rule $r=S^1\rightarrow u S^0 v$ together with index pointing to a particular position
inside this rule. More precisely, $g$ can be one of the following:
(a) $g=u_1.u_2 S^0 v$ with $u=u_1u_2$, (b) $g=uS^0 v_1 .v_2$ with $v=v_1v_2$, or (c) $g=u\dot S^0 v$.
The dot indicates the position inside $r$, and corresponds to the end of pattern $\beta$. The state $g=u\dot S^0 v$ designates that the position is strictly inside the word between $u$ and $v$.
Note that $u,v$ are words, not patterns (i.e.\ they are not associated with any interval).

Message $S\left(g,\beta\right)$ is defined as
 the minimum of the functional
that includes costs of both patterns and rules
over all partial assignments  $x=x_{1:j_{\beta}}$ under two constraints:
(i) $x=\ast\beta$ and $x\ne\ast\gamma, (\beta,\gamma)\in\Pi_{j_{\beta}}$;
(ii) $x$ can be extended to some assignment $xy$
that can be derived using rule $r$ together with
all other rules counted for $x$, and the dot in $x.y$
would correspond to the dot in $g$ (see
Fig.~\ref{fig:brackets2}).
The cost of rule $r$ is counted only if $g=uS^0v.$ \footnote{Note that similar states $g$ are also used in Earley parser~\cite{Earley}.
Thus, our algorithm can be viewed as an extension of Earley parser to \GPBs.
Unfortunately, for general grammars such extension would give algorithms whose complexity is non-polynomial in $|R|$.}.

\begin{figure}[t]
\vskip 0.05in
\small
\begin{center}
\includegraphics[scale=0.4]{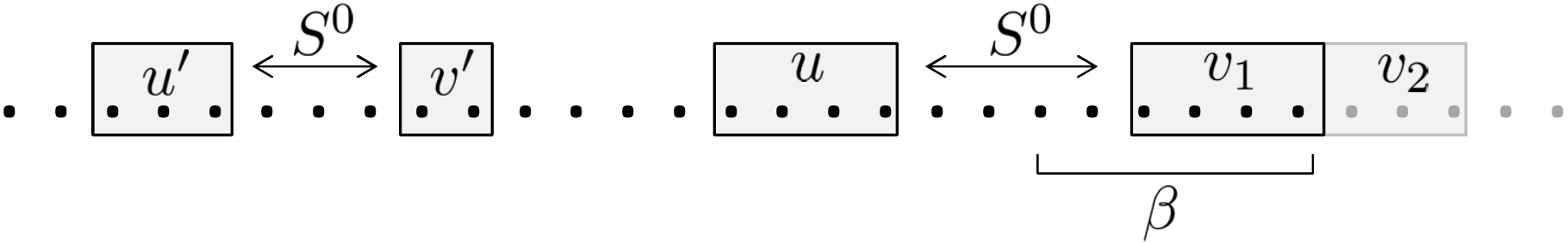} \vspace{-11pt} \\
\end{center}
\caption{A state $g=uS^0v_1.v_2, v_2\ne \varepsilon$. At a current stop at $j_{\beta}$ we have already read a sequence $u\dots v_1$ and expect $v_2$ to come. All other interacting pairs (such as pair $(u',v')$ shown in the figure)
can be located only before $u$.
}
\label{fig:brackets2}
\vskip -0.1in
\end{figure}

\begin{theorem}
\label{F1}
All messages $S\left(g,\beta\right)$, and therefore $M=\min_x F(x)$, can be computed in time $O(|P|n L(l_{\min}\min(|D|,\log l_{\min})+|P|))$.
\end{theorem}
\ifTR
\begin{proof}
We compute these messages in the order of increasing $j_{\beta}$. Let us define for a given state $g$ a set $P_{g}$ of possible states one step back (here and below $a\in D$):
$P_{u_1 a . u_2 S^0 v} = \left\{u_1 . a u_2 S^0 v\right\}$; $P_{u S^0 v_1 a . v_2} = \left\{u S^0 v_1 . a v_2\right\}$, if $v_2\ne \varepsilon$; $P_{uaS^0.v} = \left\{ua \dot S^0 v, ua.S^0 v, u.aS^0 v\right\}$; $P_{u \dot S^0 v} = \left\{u \dot S^0 v, u.S^0 v\right\}$.

If $g\ne .u S^0 v$, it is easy to see that $S\left(g,\beta\right)$ can be calculated as a function of $S\left(g',\gamma\right),g'\in P_{g}, \gamma\in\Pi_{j_{\beta}-1}$ in the same way as it is done in algorithm for minimizing \PBs \cite{TK:ICML}. Recall, that the average price of such computation is $O\left(\min(|D|,\log l_{\min})\right)$ where $l_{\min} = \min_{w\in\Lambda}|w|$ which leads to overall complexity of such computation $O(|P| l_{\min} n L\cdot \min(|D|,\log l_{\min}) )$. The only specifics is a case of $g=u S^0 va.$ for which we also have to add a weight of rule $S^1\rightarrow u S^0 v$ to an expression when we calculate $S\left(g,\beta\right)$ based on the assumption that the previous state was $u S^0 v.a$.

For $g= .u S^0 v$, first we compute $S\left(g,\beta\right)$ as a function of $S\left(g,\gamma\right),\gamma\in\Pi_{j_{\beta}-1}$, plus we should take into consideration that $S\left(g,\beta\right)\leftarrow \min S\left(g,\beta\right)$, $\min_{(u',v')\in P} S\left(u'S^0 v'.,\beta\right)$.

After adding complexities for all cases we obtain overall complexity of $O (|P|\cdot n L (l_{\min} \min(|D|,\log l_{\min})+|P|))$.
\end{proof}
\fi

\subsection{Generalizations for rules weights}
So far we assumed for simplicity that the weights $\nu(r)$ of rules $r\in R$ do not depend on the interval $[i,j]$
for which this rule is applied. In practice, however, such dependence is desirable, since e.g.\ the input substring $z_{i:j}$
may vary for different intervals. We can incorporate this dependence by introducing weight $\nu(r,i,j)$ of a rule
$r\in R$ given that we derive subword $x_{i:j}$ from its left-side nonterminal.
It is straightforward to modify algorithms to this case (without affecting the complexity):
when computing $M_A(\alpha,\beta)$, we simply need to use $\nu(r,j_\alpha+1,j_\beta)$ instead of $\nu(r)$.
The only exception is the algorithm in Theorem~\ref{F1}, since it works with different messages.
In this case we can show the following.
\begin{theorem} Suppose that weights of rules $r\in R$ with intervals $[i,j]$ satisfy the property
$\nu(r,i,j) = f(r,i)+g(r,j)$. Then $M=\min_x F(x)$ can be computed in time $O(|P|n L(l_{\min}\min(|D|,\log l_{\min})+|P|))$.
\end{theorem}
\ifTR
\begin{proof} To prove this statement we only have to redefine messages $S\left(g,\beta\right)$ in a way that if $g=uS^0 v_1 .v_2, v=v_1 v_2$ or $g=u\dot S^0 v$ or $g=u.S^0 v$, then the weight $f(S^1\rightarrow u S^0 v,i^*)$ should be present in $S\left(g,\beta\right)$, where $i^*$ is an index from which $u$ started. The calculation of messages is only slightly different from the previous and can be easily reconstructed.
\end{proof}
\fi

\section{Learning \GPB model}
We introduced a new family of probabilistic distributions that we call a grammatical pattern-based model.
This distribution is defined on a pair of objects, i.e. $p\left(x,\lambda\right) \sim \exp\{-E_\theta(x,\lambda|z)\}$, where $x$ stands for
a labeling sequence and $\lambda$ stands for a derivation of $x$ according to some grammar $\Gamma$. Suppose that $\lambda$ is a hidden variable and we need to learn the model.
We showed that minimizing the energy over both $x$ and $\lambda$ is a tractable problem. Moreover, minimizing the energy over $\lambda$ for a fixed $x$ is equivalent
to a least-weight parsing of $x$ according to grammar $\Gamma$, which can be solved in $O\left(|R|n^3\right)$ time by the standard \CYK algorithm\cite{Aho}.
Together, these two facts open a possibility to learn  \GPB models (under condition that we parameterize energy linearly with respect
to weights to be learned) by the struct-SVM with hidden variables approach \cite{Yu:2009}.

\ifTR
Learning the model with maximum likelihood approach by either gradient-based or EM-based methods \cite{nowozin-fnt2011} requires another kind of algorithms.
First of all we need algorithms for computing expressions like $\sum_{x,\lambda}\exp\{-E_\theta(x,\lambda|z)\}$ and $\sum_{\lambda}\exp\{-E_\theta(x,\lambda|z)\}$.
It can be seen that our algorithm for general \WCFG can be turned into an algorithm for computing the first sum.
Indeed, if in the definition of messages and $f,\phi$-expressions we replace ``$\min$'' with ``$\sum$'' and ``$+$'' with multiplication
(and accordingly, ``$-$'' in the algorithm is turned into division and new weights of patterns and rules are defined as $e^{-\text{old weight}}$) we obtain a valid algorithm for computing such sums. Moreover, now we can compute the matrix product in step 3 of the algorithm using some fast matrix multiplication algorithm \cite{Coppersmith:1987}, which leads to complexity $O\left(|R|(nL)^{2.376}\right)$.
Also,
 sums of the second type can be computed in time $O\left(|R|n^{2.376}\right)$.

Such sums allow computing in polynomial time marginals
required by gradient-based or EM-based methods of learning,
e.g.\ by running the algorithm independently for each marginal with appropriately modified costs.
We conjecture, however, that the marginals can be computed more efficiently by a single
computation, similar to~\cite{Ye:NIPS09,TK:ICML}. This is left as a future work.

\fi

\ifTR
\section{Experiments and discussion}\label{sec:experiments}
To support the claim the the runtime of the algorithm for interaction grammars can be better than
$O(n^3)$ (for a fixed set of patterns and interacting pairs), we present some computational results on a synthetic data.
As a subroutine for solving all-pairs shortest path problem in step 4 of Algorithm~\ref{alg:Int}
we used a code based on \cite{Demetrescu:Exp04} that we took from \url{http://www.dis.uniroma1.it/~demetres/experim/dsp/}\footnote{Note that for Demetrescu-Italiano algorithm the complexity bound of $O\left(|E^*||V|+|V|^2\log |V|\right)$ is true only if all shortest paths are unique. This condition we can satisfy by additional subtracting of some random value from interval $[0,\varepsilon']$ for sufficiently small $\varepsilon'$ from each edge weight.}.

In \GPB potential we defined $D = \left\{0,1\right\}$ and $\Lambda = D^4$, i.e. $\left|\Lambda\right| = 16$. An interaction grammar $\Gamma$ of depth 2 contained only one ``interaction'' rule $r = S\rightarrow 11 S 11$. For each pattern $\alpha = [i_{\alpha},j_{\alpha},c_{\alpha}]$ its weight $c_{\alpha}$ was taken as a uniformly distributed random value from interval $\left[0,1\right]$. A weight $\nu\left(r,i,j\right)$ was taken as a uniformly distributed random value from interval $\left[0,C\right]$, where $C$ is a parameter that we varied from 0.0 to 10 with step 0.1. We introduced a parameter $C$ to consider cases when pattern-based part is dominant in \GPB potential ($C=0.0$) and visa versa ($C=10.0$).
The length of variable chain was varied from 10 to 350 with step 10.

\begin{figure}[t]
\vskip 0.05in
\small
\begin{center}
\includegraphics[scale=0.5]{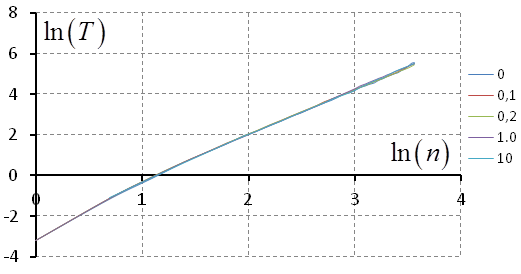}\\
\end{center}
\caption{The dependence of minimization time $T$ on chain length $n$ for different values of $C$.}
\label{fig:plot}
\end{figure}

The dependence of minimization time $T$ on chain length $n$ for different values of $C$ is shown in Fig.~\ref{fig:plot}.
For all values of $C$, in the interval  $[100,350]$ minimization time starts to depend on $n$ as $n^{x}$ with a power $x$ between 2.2 and 2.3.
Our experiments also showed that the dependence does not change whether we make subtractions of small values from edge weights that we described. A very similar dependence of $T$ on $n$ (as $n^{2.3}$) was obtained for other choices of ``interaction'' rule $r$, like e.g., $S\rightarrow 0 S 1$ or $S\rightarrow 1 S 1$.
\else
\fi

\section{Conclusions}
The \GPB model can be viewed as a natural combination of a local \PB and a more global \WCFG models:
combining different constraints by taking a sum of energy terms is a standard approach in the CRF literature.
We showed that various inference tasks in \GPBs can be solved in polynomial time.

The complexity of our general-purpose algorithm is rather high, and it can be prohibitively slow for some applications.
However, we showed there exist faster techniques in some special case, namely interaction grammars of a fixed depth.
This suggests that there may be other  classes of \GPBs with better complexity
(such as $LR(k)$ grammars).
We hope that our paper
will stimulate the search for such classes.

\bibliographystyle{plain}

\end{document}